\documentclass[letterpaper, 10 pt, conference]{ieeeconf}

\usepackage{times}  
\usepackage{helvet}  
\usepackage{courier}  
\usepackage[hyphens]{url}  
\usepackage{graphicx} 
\usepackage{caption} 
\DeclareCaptionStyle{ruled}{labelfont=normalfont,labelsep=colon,strut=off} 

\usepackage{cite}
\usepackage{booktabs}       
\usepackage{amsfonts}       
\usepackage{amssymb}
\usepackage{amsmath}
\usepackage{nicefrac}       
\usepackage{microtype}      
\usepackage{subcaption}
\usepackage{algorithm, algpseudocode}
\usepackage{amstext} 
\usepackage{array}   
\newcolumntype{L}{>{$}l<{$}} 
\newcolumntype{C}{>{$}c<{$}} 
\usepackage{siunitx} 
\usepackage{multirow}

\usepackage{wrapfig,lipsum,booktabs}

\newtheorem{assumption}{Assumption}
\newtheorem{problem}{Problem}

\newcommand{\R}{\mathbb{R}}
\newcommand{\bmat}[1]{\begin{bmatrix}#1\end{bmatrix}}
\newcommand{\smat}[1]{\left[\begin{smallmatrix} #1 \end{smallmatrix} \right]}

\newtheorem{lemma}{Lemma}
\newtheorem{theorem}{Theorem}
\newtheorem{definition}{Definition}

\usepackage{newfloat}
\usepackage{listings}
\floatstyle{ruled}
\newfloat{listing}{tb}{lst}{}
\floatname{listing}{Listing}
\usepackage{hyperref}       
\hypersetup{
    colorlinks = true,
    citecolor = {blue}}

\IEEEoverridecommandlockouts 
\title{\LARGE \bf 
Synthesis of Stabilizing Recurrent Equilibrium Network Controllers}
\author{
     Neelay Junnarkar, He Yin, Fangda Gu, Murat Arcak, Peter Seiler 
\thanks{Funded in part by the Air Force Office of Scientific Research grant FA9550-21-1-0288, and the Office of Naval Research grant N00014-18-1-2209.}
\thanks{N. Junnarkar, H. Yin, F. Gu, and M. Arcak are with the University of California, Berkeley {\tt\small \{neelay.junnarkar, he\_yin, gfd18, arcak\}@berkeley.edu.}}
\thanks{P. Seiler is with the University of Michigan,  Ann Arbor {\tt\small pseiler@umich.edu.}}
}

\begin{document}

\maketitle
\thispagestyle{empty}
\pagestyle{empty}


\begin{abstract}

We propose a parameterization of a nonlinear dynamic controller based on the 
recurrent equilibrium network, a generalization of the recurrent neural network. 
We derive constraints on the parameterization under which the controller 
guarantees exponential stability of a partially observed dynamical system with 
sector bounded nonlinearities. 
Finally, we present a method to synthesize this controller using projected policy
gradient methods to maximize a reward function with arbitrary structure. 
The projection step involves the solution of convex optimization problems. 
We demonstrate the proposed method with simulated examples of controlling a nonlinear inverted pendulum and a plant modeled with neural networks.

\end{abstract}

\section{Introduction}

Neural networks (NNs) in control tasks have seen recent success, particularly 
through reinforcement learning, which allows for application to control problems 
with complex reward functions or unknown dynamics  \cite{sutton2018reinforcement}.
However, many reinforcement learning methods focus on achieving desirable criteria
such as stability by encoding these goals into a reward function which the NN is 
trained to maximize.
This does not enforce that a trained controller will guarantee the desired 
properties.
In safety critical systems, it is preferable to apply properties such as stability
as a constraint. 
Work in \cite{kretchmar2001} develops methods for provably stable NN controllers 
with reinforcement learning methods; \cite{anderson2007robust} uses recurrent 
neural network (RNN) controllers for guaranteed stabilization of partially 
observed linear systems. Reference \cite{gu2021recurrent} extends 
\cite{anderson2007robust} to uncertain nonlinear systems, and develops convex inner-approximations to the set of stabilizing controller parameters.

Recent work in \cite{elghaoui2020} and \cite{bai2019deep} develop an expressive 
class of NNs in which the output of a layer is defined implicitly. 
The recurrent equilibrium network (REN), a model which leverages these implicit 
properties, is presented in \cite{revay2021ren}. 
References \cite{wang2021yoularen} and \cite{wang2021LearningOA} 
develop REN-based state- and output-feedback controllers with guarantees of stabilizing 
linear systems.

We consider the problem of designing a controller that maximizes a reward function with arbitrary structure while guaranteeing stability of a partially-observed system with sector-bounded nonlinearities. 
While stability is guaranteed, the reward function can be used to promote other 
desirable properties, e.g., by acting as a soft constraint on the control action 
magnitude \cite{wang2021LearningOA}. 
Partially-observed systems are common in practice due to sensing limitations 
\cite{braziunas2003pomdp} and estimates of the full state often require historical
output data \cite{callier1991linear}.
Sector bounded nonlinearities can describe numerous kinds 
of nonlinear functions, e.g., activation functions 
\cite{anderson2007robust,  fazlyab2020safety,pauli2021training}, saturation 
functions~\cite{hindi1998analysis}, and model uncertainties \cite{buch2021robust}.

\subsubsection{Contributions}

In this paper we use a model based on the REN to create a class of nonlinear 
dynamic controllers.
We derive a Lyapunov condition on the parameters of the REN-based controller under which 
the controller guarantees exponential stability of the origin of a partially observed dynamical 
system with sector bounded nonlinearities.
Based on the techniques of loop transformation, and change of variables adapted from
\cite{scherer1997multiobjective}, we derive a convex approximation to the set of 
stabilizing parameters. 
In the case that the plant is a linear system, this convexification procedure is 
lossless, i.e. the convex stability condition is equivalent to the non-convex one. 
We demonstrate a method to synthesize this controller using gradient-based 
reinforcement learning algorithms combined with a convex projection step.
The effectiveness of this controller parameterization is demonstrated on a 
nonlinear inverted pendulum model.
We also include an example of stabilizing a plant where the plant model has been learned as an implicit NN.

Compared with \cite{anderson2007robust, gu2021recurrent}, this paper extends the 
framework from explicit RNN controllers to RENs, enriching the description of 
controllers. 
Unlike the convex inner approximation presented in \cite{gu2021recurrent}, the convex stability condition presented here for linear plants is also lossless.
Additionally, in contrast to the Youla parameterization method in 
\cite{wang2021yoularen, wang2021LearningOA}, this paper does not require a ``base'' output feedback 
linear controller, and allows for nonlinear plants, including NN modeled plants.
An example with plants modeled as NNs is included at \url{https://arxiv.org/abs/2204.00122}.

\subsubsection{Notation}    
$\mathbb{S}^n_{++}$ denotes the set of $n$-by-$n$ 
symmetric, positive definite matrices. 
$\mathbb{D}^n_{+}$, $\mathbb{D}^n_{++}$ denote the set of $n$-by-$n$ diagonal 
positive semidefinite, and diagonal positive definite matrices. 
The notation $\|\cdot\|: \R^n \rightarrow \R$ denotes the standard 2-norm. 
The notation $x^*$ indicates the conjugate transpose of the vector $x$.
We define $\ell_{2e}^n$ to be the set of all one-sided sequences $x: \mathbb{N} 
\rightarrow \R^n$. 
When applied to vectors, the orders $>, \leq$ are interpreted elementwise.

\section{Partially Observed Linear Systems} \label{sec:LTI}
\subsection{Problem Formulation}
Consider the feedback system consisting of a plant $G$ and a controller $\pi_\theta$ which must stabilize the origin. 
To streamline the presentation, we first consider a partially observed linear, time-invariant (LTI) system~$G$ defined by the following discrete-time model:
\begin{subequations}\label{eq:nomi_G}
\begin{align}
x(k+1) &= A_G \ x(k) + B_G \ u(k) \\
y(k) &= C_G \ x(k)
\end{align}
\end{subequations}
where $x(k) \in \R^{n_G}$ is the state, $u(k) \in \R^{n_u}$ is the control input,  $y(k) \in \R^{n_y}$ is the output, $A_G \in \R^{n_G \times n_G}$, $B_G \in \R^{n_G \times n_u}$, and $C_G \in \R^{n_y \times n_G}$.  
\begin{assumption}
We assume that $(A_G, B_G)$ is stabilizable, and $(A_G, C_G)$ is detectable \cite{callier1991linear}.
\end{assumption}

\begin{assumption}\label{ass:known_dyn}
We assume $A_G, B_G,$ and $C_G$ are known.
\end{assumption}

\begin{problem}
Our goal is to synthesize a dynamic controller $\pi_\theta$ that maps the observation $y$ to an action $u$ to both stabilize $G$ and maximize some reward $R = \sum_{k=0}^T r_k(x(k), u(k))$ over finite horizon $T$.
\end{problem}

The single step reward $r_k(x(k), u(k))$ is assumed to be unknown, i.e. we only require access to an oracle which can evaluate \(r_k(x(k), u(k))\), and may be highly complex to capture a variety of desired closed-loop properties. 
For example, a safety violation at step $l$ can be encoded by setting the single step rewards $r_k(x(k), u(k)) = 0$ for \(k \geq l\).
This type of reward cannot be represented by the more common (and more analytically tractable) negative quadratic reward function.

\subsection{Controller Parameterization} \label{sec:controller}
Although the plant considered in the section is linear, the optimal controller for a general non-quadratic reward function can be nonlinear. 
We consider a nonlinear dynamic controller based on the recurrent equilibrium network (REN) \cite{wang2021yoularen}, which makes a class of high-capacity flexible controllers. 

The controller \(\pi_\theta\) is modeled as an interconnection of an LTI system $P_\pi$ and activation functions $\phi: \R^{n_\phi} \rightarrow \R^{n_\phi}$ as shown in Fig.~\ref{fig:RNN}. 
This parameterization is expressive, and contains many widely used model structures. 
The controller $\pi_\theta$ is defined as follows 
\begin{align} 
&P_\pi \left\{\begin{array}{ll}
\xi(k+1) &= A_K \ \xi(k) \ + B_{K1} \ w(k) + B_{K2} \ y(k) \\
u(k) &= C_{K1} \ \xi(k) + D_{K1} \ w(k) + D_{K2}  \ y(k) \\
v(k) &= C_{K2} \ \xi(k) + D_{K3} \ w(k) + D_{K4} \ y(k) \end{array}\right. \notag \\
&\hspace{0.87cm} w(k) \hspace{0.85cm} = \phi(v(k)) \label{eq:PK_def}
\end{align}
where $\xi \in \R^{n_\xi}$ is the hidden state, $v, w \in \R^{n_\phi}$ are the input and output of $\phi$, and matrices $A_K, \dots, D_{K4}$ are to be learned. 
Note that \(w(k)\) is defined implicitly as the solution of \(w(k) = \phi(C_{K2} \xi(k) + D_{K3}w(k) + D_{K4}y(k))\). 
This \(w(k)\) can be interpreted as the output of an infinite layer feed-forward neural network with activation function \(\phi\) and weights \(D_{K3}\) at each layer, with an input \(C_{K2}\xi(k) + D_{K4}y(k)\) being fed in at each layer \cite{bai2019deep}. 
When \(D_{K3} = 0\), \(w(k)\) can be computed explicitly as a function of \(v(k)\) and this parameterization reduces to the RNN parameterization presented in \cite{gu2021recurrent}.

\begin{figure}
  \centering
  \includegraphics[width=0.7\linewidth]{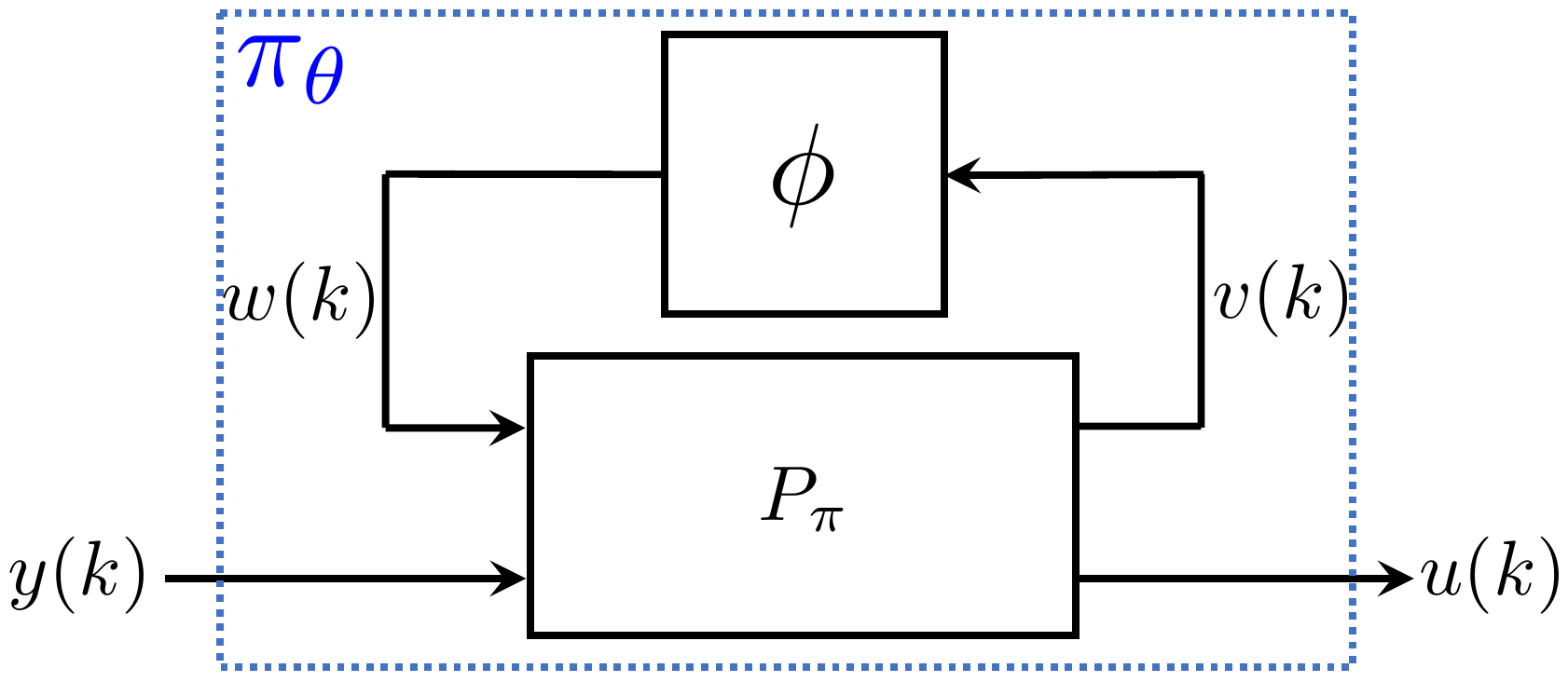}
  \caption{Controller \(\pi_\theta\) as an interconnection of $P_\pi$ and $\phi$.}
  \label{fig:RNN}   
\end{figure}

Define $\theta = \smat{A_K & B_{K1} & B_{K2} \\ C_{K1} & D_{K1} & D_{K2} \\ C_{K2} & D_{K3} & D_{K4}}$ as the collection of the learnable parameters of $\pi_\theta$. 
Let the initial condition of $\xi$  be zero: $\xi(0) = 0_{n_\xi \times 1}$. 
The combined nonlinearity $\phi$ is applied element-wise, i.e., $\phi := [\varphi_1(v_1), ..., \varphi_{n_\phi}(v_{n_\phi})]^\top$, where $\varphi_i$ is the $i$-th scalar activation function. Let each \(\varphi_i\) be differentiable.
\begin{assumption}\label{ass:nonlinearity}
We assume that each scalar nonlinearity is sector bounded, so that for \(i = 1,\dots,n_\phi\), \(\forall x\in\mathbb{R},\ \alpha_{\phi,i}x \leq \varphi_i(x) \leq \beta_{\phi, i}x\).
\end{assumption}

\subsection{Quadratic Constraints for Activation Functions}
The stability condition relies on quadratic constraints (QCs) \cite{fazlyab2020safety, megretski1997system} to bound the activation function. 
The sector bound on the scalar nonlinearities can be expressed as a QC.
\begin{definition}
  \label{def:sector}
  Let $\alpha \le \beta$ be given. 
  The function $\varphi: \R \rightarrow \R$ lies in the sector $[\alpha,\beta]$ if:
  \begin{align}
    ( \varphi(\nu) - \alpha \nu ) \cdot
       (\beta \nu - \varphi(\nu)) \ge 0
    \,\,\, \forall \nu \in \R.
  \end{align}
\end{definition}

The interpretation of the sector $[\alpha,\beta]$ is that $\varphi$ lies between lines passing through the origin with slope $\alpha$ and $\beta$. 
Many activations are sector bounded: leaky ReLU is sector bounded in $[a, 1]$ with its parameter $a \in (0, 1)$; ReLU and $\tanh$ are sector bounded in $[0,1]$ (denoted as $\tanh \in $ sector $[0, 1]$). 
Fig.~\ref{fig:sector} illustrates the sector bound on \(\tanh\).
\begin{figure}[b]
    \centering
    \includegraphics[width=0.7\linewidth,trim={0 0 0 0.17cm},clip]{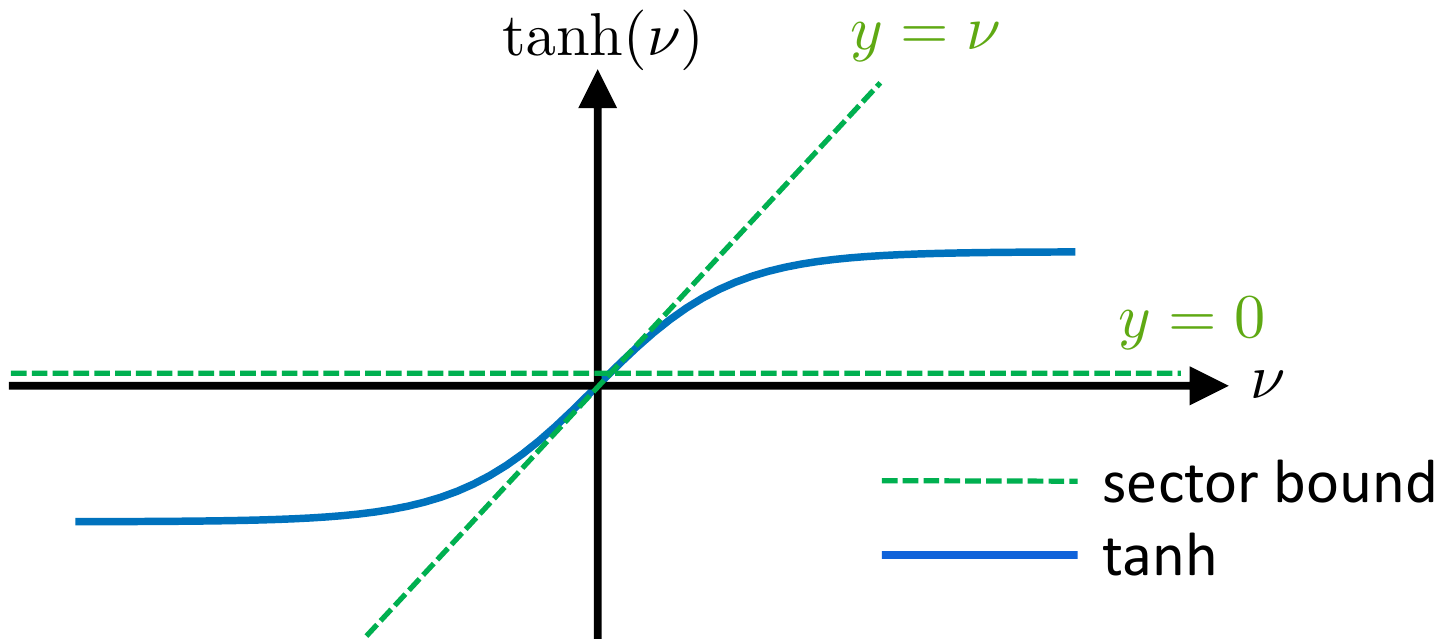}
    \caption{\(\tanh \in\) sector \([0, 1]\).}
    \label{fig:sector}
\end{figure}

Sector constraints can also be defined for combined activations $\phi$. 
From Assumption~\ref{ass:nonlinearity}, the \(i\)-th scalar activation function \(\varphi_i\) is sector bounded in \([\alpha_{\phi, i}, \beta_{\phi, i}]\).
These sectors can be stacked into vectors $\alpha_\phi, \beta_\phi \in \R^{n_\phi}$, where $\alpha_\phi=[\alpha_{\phi,1},...,\alpha_{\phi,n_\phi}]$ and $\beta_\phi=[\beta_{\phi,1}, ..., \beta_{\phi,n_\phi}]$, to provide QCs satisfied by  $\phi$.
\begin{lemma}
Let $\alpha_\phi, \beta_\phi \in \R^{n_\phi}$ be given with $\alpha_\phi \leq \beta_\phi$. 
Suppose that $\phi$ satisfies the sector bound $[\alpha_\phi, \beta_\phi]$ element-wise. 
For any $\Lambda \in \mathbb{D}_{+}^{n_\phi}$, and for all $v \in \R^{n_\phi}$ and $w = \phi(v)$, it holds that
\begin{align}
        \bmat{v \\ w}^\top \bmat{-2A_\phi B_\phi \Lambda & (A_\phi+B_\phi)\Lambda \\ (A_\phi+B_\phi)\Lambda & -2\Lambda} \bmat{v \\ w} \ge 0, \label{eq:origin_QC}
\end{align}
where $A_\phi = \mathrm{diag}(\alpha_\phi)$, and $B_\phi = \mathrm{diag}(\beta_\phi)$.
\end{lemma}
A proof is available in \cite{fazlyab2020safety}.


\subsection{Loop Transformation} \label{sec:loop-transform}
We first perform a loop transformation
to simplify the non-convex stability condition that serves as a basis for the convex stability condition.
This transformation yields a new representation of the controller $\pi_{\tilde{\theta}}$, which is equivalent to the one shown in Fig.~\ref{fig:RNN}.


Define \(S_\phi = \frac{1}{2}(A_\phi + B_\phi)\) and \(L_\phi = \frac{1}{2}(B_\phi - A_\phi)\).
Then the loop transformed $\tilde{\phi}$ is defined by \(\tilde{\phi}(v) = L_\phi^{-1}(\phi(v) - S_\phi v)\). 
This transforms \(\phi\), sector-bounded by \([\alpha_\phi, \beta_\phi]\), to \(\tilde{\phi}\), sector-bounded by $[-1_{n_\phi \times 1}, 1_{n_\phi \times 1}]$.
Thus, \(\tilde{\phi}\) satisfies the following simplified QC: for any $\Lambda \in \mathbb{D}_{+}^{n_\phi}$, it holds that
\begin{align}
    \bmat{v \\ z}^\top \bmat{\Lambda & 0 \\ 0 & -\Lambda} \bmat{v \\ z} \ge 0, \ \forall v \in \R^{n_\phi} \ \text{and} \ z = \tilde{\phi}(v). \label{eq:shifted_QC}
\end{align}

The transformed system $\tilde{P}_\pi$ is:
\begin{subequations}
\begin{align}
    \xi(k+1) &= \tilde{A}_K ~~ \xi(k) + \tilde{B}_{K1} \ z(k) + \tilde{B}_{K2} \ y(k) \\
u(k) &= \tilde{C}_{K1} \ \xi(k) + \tilde{D}_{K1} \ z(k) + \tilde{D}_{K2} \ y(k) \\
v(k) &= \tilde{C}_{K2} \ \xi(k) + \tilde{D}_{K3} \ z(k) + \tilde{D}_{K4} \ y(k) \\
z(k) &= \tilde{\phi}(v(k)) \label{eq:implicit-z-eq}
\end{align}\label{eq:tilde-controller}
\end{subequations}

where $M_{K3} = (I - S_\phi D_{K3})^{-1}$,
{\small
\begin{align}
\tilde{A}_K &= A_K + B_{K1} M_{K3} S_\phi C_{K2}, \ \tilde{B}_{K1} = B_{K1} M_{K3} L_\phi,\label{eq:1to1correspond} \\
\tilde{B}_{K2} &= B_{K2} + B_{K1} M_{K3} S_\phi D_{K4}, \nonumber\\
\tilde{C}_{K1} &= C_{K1} + D_{K1} M_{K3} S_\phi C_{K2}, \ \tilde{D}_{K1} = D_{K1} M_{K3} L_\phi \nonumber\\
\tilde{D}_{K2} &= D_{K2} + D_{K1} M_{K3} S_\phi D_{K4}, \nonumber \\
\tilde{C}_{K2} &= C_{K2} + D_{K3} M_{K3} S_\phi C_{K2}, \ \tilde{D}_{K3} = D_{K3} M_{K3} L_\phi \nonumber\\
\tilde{D}_{K4} &= D_{K4} + D_{K3} M_{K3} S_\phi D_{K4} . \nonumber 
\end{align}
}
The derivation of $\tilde{P}_\pi$ can be found in \cite{gu2021recurrent}. 
We define the learnable parameters of $\pi_{\tilde{\theta}}$ as $\tilde{\theta} = \smat{\tilde{A}_K & \tilde{B}_{K1} & \tilde{B}_{K2} \\ \tilde{C}_{K1} & \tilde{D}_{K1}& \tilde{D}_{K2} \\  \tilde{C}_{K2} & \tilde{D}_{K3} & \tilde{D}_{K4}}$.
Since there is a one-to-one correspondence \eqref{eq:1to1correspond} between the transformed parameters $\tilde{\theta}$ and the original parameters $\theta$, we learn and evaluate the controller in the reparameterized space.

\subsection{Convex Stability Condition}
The feedback system of plant $G$ and REN controller in $\pi_{\tilde{\theta}}$
is defined by the following equations
\begin{subequations}\label{eq:feedback_nominal}
\begin{align}
    \zeta(k+1) &= \mathcal{A} \ \zeta(k) + \mathcal{B} \ z(k)\\
v(k) &= \ \mathcal{C} \ \zeta(k) + \mathcal{D} \ z(k) \\
z(k) & = \tilde{\phi}(v(k)) \label{eq:implicit}
\end{align}
\end{subequations}
where $\zeta = [x^\top, \ \xi^\top] ^\top$ collects the states of $G$ and $\pi_{\tilde{\theta}}$ and
\begin{align*}
\mathcal{A} &= \smat{ A_G + B_G \tilde{D}_{K2}C_G & B_G \tilde{C}_{K1} \\ \tilde{B}_{K2}C_G & \tilde{A}_{K}}, \mathcal{B} = \smat{B_G \tilde{D}_{K1} \\ \tilde{B}_{K1}},\\
\mathcal{C} &= \smat{\tilde{D}_{K4}C_G & \tilde{C}_{K2}}, \hspace{1.58cm} \mathcal{D} = \tilde{D}_{K3}.
\end{align*}

The following lemma incorporates the QC for $\tilde{\phi}$ in the Lyapunov condition to derive the exponential stability condition of the origin using the S-Lemma \cite{boyd1994linear}.

\begin{lemma} \label{lemma:lin-plant-stability}
Consider the feedback system of plant G in \eqref{eq:nomi_G} and controller \(\pi_{\tilde{\theta}}\) in \eqref{eq:tilde-controller}, with the transformed nonlinearity \(\tilde{\phi}\) being sector bounded by \([-1, 1]\).
Assume that the implicit equation for \(z\) in the controller is well-posed, i.e., for any \(\zeta\) there 
exists a unique solution \(z\) to \eqref{eq:implicit}.
Given a rate \(\rho \in [0, 1)\) and parameters \(\tilde{\theta}\), if there exist matrices \(P \in \mathbb{S}_{++}^{n_G + n_\xi}\) and \(\Lambda \in \mathbb{D}_{++}^{n_\phi}\) such that the following condition holds
\begin{align}\label{eq:lyap_cond}
\begin{bmatrix}\mathcal{A}^\top P \mathcal{A}- \rho^2 P & \mathcal{A}^\top P \mathcal{B} \\ \mathcal{B}^\top P \mathcal{A} & \mathcal{B}^\top P \mathcal{B} \end{bmatrix}  +  [\star]^\top \begin{bmatrix}\Lambda & 0 \\ 0 & -\Lambda \end{bmatrix}\begin{bmatrix}\mathcal{C} & \mathcal{D} \\ 0 & I \end{bmatrix}  \prec 0,
\end{align}
where \(\star\) is inferred from symmetry, then for any \(x(0)\), we have \(\|x(k)\| \leq \sqrt{\mathrm{cond}(P)}\rho^k\|x(0)\|\) for all \(k > 0\) where \(\mathrm{cond}(P)\) is the condition number of \(P\) under the \(L^2\) norm.
The origin is exponentially stable with rate \(\rho\).
\end{lemma}
\begin{proof}
    The proof of stability is adapted from \cite{gu2021recurrent}.
\end{proof}

The main conditions in the lemma above are Assumption~\ref{ass:nonlinearity} (sector boundedness of the nonlinearity), well-posedness of the controller, and existence of \(P\) and \(\Lambda\) such that \eqref{eq:lyap_cond} holds. 
If we replace Assumption~\ref{ass:nonlinearity} with the stricter condition that each scalar nonlinearity \(\varphi_i\) is slope restricted and passes through the origin, i.e., such that \(\varphi_{i}^{\prime} \in [\alpha_{\phi, i}, \beta_{\phi, i}]\) and \(\varphi_i(0) = 0\), then \(\varphi_i\) is sector bounded by \([\alpha_{\phi, i}, \beta_{\phi, i}]\). 
Further, when combined with existence of \(P\) and \(\Lambda\) such that \eqref{eq:lyap_cond} holds, slope restriction of \(\varphi_i\) implies well-posedness of the controller. 
Therefore, in Lemma~\ref{lemma:lin-plant-stability}, Assumption~\ref{ass:nonlinearity} and well-posedness can be replaced by the condition that the scalar nonlinearities be slope restricted and pass through the origin.

\begin{lemma} \label{lemma:well-posedness}
    Suppose that for \(i = 1,\dots,n_\phi\) the scalar nonlinearity \(\varphi_i\) satisfies \(\varphi_i^\prime \in [\alpha_{\phi, i}, \beta_{\phi, i}]\) and \(\varphi_i(0) = 0\).
    Also assume that for given \(\rho \in [0, 1)\) and parameters \(\tilde{\theta}\), there exist matrices \(P \in \mathbb{S}_{++}^{n_G + n_\xi}\) and \(\Lambda \in \mathbb{D}_{++}^{n_\phi}\) such that \eqref{eq:lyap_cond} holds.
    Then for any \(\zeta\), there exists a unique solution \(z\) to \(z = \tilde{\phi}(C\zeta + Dz)\), i.e. the controller is well-posed.

\end{lemma}

\begin{proof}
Define the vector norm \(\|x\|_\Lambda \triangleq \|\Lambda^{\frac{1}{2}}x\|_2\).
This is well-defined since \(\Lambda \in \mathbb{D}_{++}\).
This induces the following matrix norm:
\[\|\mathcal{D}\|_\Lambda \triangleq \sup_{x \neq 0} \frac{\|Dx\|_\Lambda}{\|x\|_\Lambda}.\]
From the (2, 2) block of \eqref{eq:lyap_cond} we infer that
\begin{align*}
    \mathcal{B}^\top P \mathcal{B} + \mathcal{D}^\top \Lambda \mathcal{D} - \Lambda \prec 0.
\end{align*}
Since \(P \succ 0\), we have \(\mathcal{D}^\top \Lambda \mathcal{D} - \Lambda \prec 0\).
Therefore, \(\|Dx\|_\Lambda^2 = x^* \mathcal{D}^\top \Lambda \mathcal{D} x < x^* \Lambda x = \|x\|_\Lambda^2\), implying \(\|\mathcal{D}\|_\Lambda <  1\).

For a fixed \(\zeta\), the Jacobian of \(\tilde{\phi}(\mathcal{D}z + \mathcal{C}\zeta)\) with respect to \(z\) is the matrix  \(\tilde{\phi}^\prime(\mathcal{D}z+\mathcal{C}\zeta)\mathcal{D}\). 
The derivative \(\tilde{\phi}^\prime\) is diagonal since \(\phi\) is assumed to operate element-wise.
Further, loop transformation ensures \(\tilde{\varphi}_i^\prime \in [-1, 1]\), so each diagonal entry of \(\tilde{\phi}^\prime\) is in \([-1, 1]\). 
Therefore, \(\tilde{\phi}^{\prime \top} \Lambda \tilde{\phi}^\prime \preceq \Lambda\), so \(\|\tilde{\phi}^\prime\|_\Lambda \preceq 1\).

Combining the two norm bounds, \( \|\tilde{\phi}^\prime(\mathcal{D}z+\mathcal{C}\zeta)\mathcal{D}\|_\Lambda \leq \|\tilde{\phi}^\prime(\mathcal{D}z+\mathcal{C}\zeta)\|_\Lambda \|\mathcal{D}\|_\Lambda < 1\). 
Since the norm of the Jacobian is less than 1, the mapping \(\tilde{\phi}(\mathcal{D}z + \mathcal{C}\zeta)\) is a contraction with respect to \(z\).
Therefore there exists a unique solution \(z(k)\) for each \(\zeta(k)\) in \eqref{eq:tilde-controller} \cite{sastry2013nonlinear}, i.e. \(\pi_{\tilde{\theta}}\) is well-posed.
\end{proof}

Note that matrices \(\mathcal{A}, \mathcal{B}, \mathcal{C}\), and \(\mathcal{D}\) depend on \(\tilde{\theta}\). 
Hence, \eqref{eq:lyap_cond} is not convex in \(\tilde{\theta}\), \(P\), and \(\Lambda\).
We now construct an equivalent stability condition that is convex so that projecting \(\tilde{\theta}\) to the set of stabilizing parameters can be solved with a convex optimization problem.

Rearranging \eqref{eq:lyap_cond} and using \(P \succ 0\) and \(\Lambda \succ 0\), the Schur complement gives the following equivalent condition:
\begin{equation}\label{eq:big_lyap_cond}
\bmat{ \rho^2 P & 0 & \mathcal{A}^\top & \mathcal{C}^\top \\
0 & \Lambda & \mathcal{B}^\top & \mathcal{D}^\top \\
\mathcal{A} & \mathcal{B} & P^{-1} & 0 \\
\mathcal{C} & \mathcal{D} & 0 & \Lambda^{-1}
} \succ 0.
\end{equation}
Multiply \eqref{eq:big_lyap_cond} on the left and right by $\mathrm{diag}(I, I, P, \Lambda)$, which is invertible, to get
\begin{equation} \label{eq:big_lyap_cond_multiply}
\bmat{ \rho^2 P & 0 & \mathcal{A}^\top P & \mathcal{C}^\top \Lambda \\
0 & \Lambda & \mathcal{B}^\top P & \mathcal{D}^\top \Lambda\\
P \mathcal{A} & P \mathcal{B} & P & 0 \\
\Lambda \mathcal{C} & \Lambda \mathcal{D} & 0 & \Lambda
} \succ 0.
\end{equation}

According to the partition of $\mathcal{A}$, we introduce the following notations for the sub-blocks of $P$ and its inverse:
\begin{equation}
    P = \bmat{X & U \\ U^\top & \hat{X}} ,  \ P^{-1} = \bmat{Y & V \\ V^\top & \hat{Y}}, \label{eq:P-partition}
\end{equation}
and thus $YX + V U^\top = I$, where \(X,Y \in \R^{n_G \times n_G}, U,V \in \R^{n_G \times n_\xi}\), and \(\hat{X}, \hat{Y} \in \R^{n_\xi, n_\xi}\). This partition was previously introduced by \cite{scherer1997multiobjective} in the LMI solution of an \(H_\infty\) synthesis problem.

Define a matrix $\mathcal{Y}$ as
\begin{equation} \label{eq:mathcaly-def}
    \mathcal{Y} \triangleq \bmat{Y & I \\ V^\top & 0}.
\end{equation}

Note this matrix has full column rank and is invertible when the state size of \(\pi_{\tilde{\theta}}\) equals the plant state size.
Multiply \eqref{eq:big_lyap_cond_multiply} on the left by $diag(\mathcal{Y}^\top, I, \mathcal{Y}^\top, I)$ and on the right by its transpose to get the the following condition:
\begin{equation} \label{eq:big_lyap_cond_multiply2}
\bmat{ \rho^2 \mathcal{Y}^\top P \mathcal{Y} & 0 & \mathcal{Y}^\top \mathcal{A}^\top P \mathcal{Y}& \mathcal{Y}^\top \mathcal{C}^\top \Lambda \\
0 & \Lambda & \mathcal{B}^\top P \mathcal{Y} & \mathcal{D}^\top \Lambda\\
\mathcal{Y}^\top P \mathcal{A} \mathcal{Y} & \mathcal{Y}^\top P \mathcal{B} & \mathcal{Y}^\top P \mathcal{Y} & 0 \\
\Lambda \mathcal{C} \mathcal{Y} & \Lambda \mathcal{D} & 0 & \Lambda
} \succ 0.
\end{equation}

While the terms in \eqref{eq:big_lyap_cond_multiply2} are not convex in \(\tilde{\theta}\), they can be expressed linearly in terms of a new set of decision variables. Specifically, define:
\begin{align*}
    N & \triangleq \smat{X A_G Y & 0 \\ 0 & 0}  + \smat{U & X B_G \\ 0 & I} \smat{\tilde{A}_K & \tilde{B}_{K2} \\ \tilde{C}_{K1} & \tilde{D}_{K2}} \smat{V^\top & 0 \\ C_G Y & I} \\
    &=\smat{N_{11} & N_{12} \\ N_{21} & N_{22}}, \\
    \hat{N}_{12} & \triangleq X B_G \tilde{D}_{K1} + U\tilde{B}_{K1}, \\
    \hat{N}_{21} & \triangleq \hat{D}_{K4} C_G Y + \Lambda \tilde{C}_{K2} V^\top, \\
    \hat{D}_{K4} & \triangleq \Lambda \tilde{D}_{K4}, \\
    \hat{D}_{K3} & \triangleq \Lambda \mathcal{D}.
\end{align*}

Then,
\begin{align*}
    \mathcal{Y}^\top P \mathcal{Y} & = \bmat{Y& I \\ I & X}, \\
    \mathcal{Y}^\top P \mathcal{AY} & = \bmat{
        A_G Y + B_G N_{21} & A_G + B_G N_{22} C_G \\
        N_{11} & X A_G + N_{12} C_G
    }, \\
    \mathcal{Y}^\top P\mathcal{B} & = \bmat{
        B_G \tilde{D}_{K1} \\
        \hat{N}_{12}
    }, \\
    \Lambda \mathcal{CY} & = \bmat{
        \hat{N}_{21} & \hat{D}_{K4} C_G
    }, \\
    \Lambda \mathcal{D} & = \hat{D}_{K3}.
\end{align*}

With this expansion, \eqref{eq:big_lyap_cond_multiply2} is linear, and therefore convex, in the decision variables $\hat{\theta}\triangleq$ ($X$, $Y$, $N$, $\Lambda$, $\hat{N}_{12}$, $\hat{N}_{21}$, $\tilde{D}_{K1}$, $\hat{D}_{K3}$, $\hat{D}_{K4}$).
Denote  \eqref{eq:big_lyap_cond_multiply2} as LMI$(\hat{\theta})$, and define the stabilizing set of parameters
\begin{align}
    \Theta = \{\hat{\theta}: \text{LMI$(\hat{\theta})$ \ holds}\}.
\end{align}

\begin{theorem} \label{thm:stability-lin}
Consider the feedback system of plant $G$ in \eqref{eq:nomi_G} and well-posed controller \(\pi_{\tilde{\theta}}\) in \eqref{eq:tilde-controller}.
Let a rate \(\rho \in [0, 1)\) be given.
If there exists \(\hat{\theta}\) such that \(\mathrm{LMI}(\hat{\theta})\) holds, then there exist parameters \(\tilde{\theta}\) for \(\pi_{\tilde{\theta}}\) such that origin is exponentially stable with rate \(\rho\).
\end{theorem}

\begin{proof}
    Given \(\hat{\theta} \in \Theta\), construct \(U\) and \(V\) such that \(UV^\top = I - XY\), e.g. through SVD. 
    Assuming \(n_\xi = n_G\), one solution is \(U = X\) and \(V = X^{-1} - Y\). Construct \(P\) as 
    \begin{equation}
        \left(\mathcal{Y}^\top\right)^{-1} \bmat{Y & I \\ I & X} \mathcal{Y}^{-1}.
    \end{equation}
    Recover \(\tilde{\theta}\) as follows:
    \begin{subequations}\label{eq:recovery}
    \begin{align}
        \smat{\tilde{A}_K & \tilde{B}_{K2} \\ \tilde{C}_{K1} & \tilde{D}_{K2}} &= \smat{U & X B_G \\ 0 & I}^{-1} \left(N - \smat{X A_G Y & 0 \\ 0 & 0} \right) \smat{V^\top & 0 \\ C_G Y & I}^{-1}, \\
        \tilde{B}_{K1} &= U^{-1} (\hat{N}_{12} - X B_G \tilde{D}_{K1}), \\
        \tilde{C}_{K2} &= \Lambda^{-1} (\hat{N}_{21} - \hat{D}_{K4} C_G Y) \left(V^\top\right)^{-1}, \\
        \tilde{D}_{K3} &= \Lambda^{-1}\hat{D}_{K3}, \\
        \tilde{D}_{K4} &= \Lambda^{-1}\hat{D}_{K4}. 
    \end{align}
    \end{subequations}
    Since all congruence transformations were invertible, they can be undone to arrive at \eqref{eq:big_lyap_cond} from \(\mathrm{LMI}(\hat{\theta})\). 
    Since \eqref{eq:big_lyap_cond} is equivalent to the condition in Lemma \ref{lemma:lin-plant-stability} by Schur complement, then \(P\) and \(\Lambda\) certify that the recovered \(\tilde{\theta}\) exponentially stabilize the origin with rate \(\rho\).
\end{proof}

\subsection{Stabilizing Reinforcement Learning Problem}
Denote the relationship between $\hat{\theta}$ and $\tilde{\theta}$ provided in \eqref{eq:recovery} concisely as 
\begin{align}
    \tilde{\theta} = f(\hat{\theta}).
\end{align}
The stabilizing reinforcement learning problem is to synthesize a stabilizing controller, parameterized by \(\tilde{\theta}\), to maximize the reward \(R\). This can be written as the following constrained optimization problem:
\begin{subequations} \label{eq:stab_RL_prob}
\begin{align}
    &\max_{\hat{\theta}} R\left(\pi_{f(\hat{\theta})}\right) \\
    &\text{s.t. \ LMI$(\hat{\theta})$ holds}.
\end{align}
\end{subequations}

We use gradient-based methods to optimize \(\hat{\theta}\). Computing the gradient of the reward with respect to \(\hat{\theta}\) requires the existence of the matrix inverse in the following Jacobian: 
\begin{equation*}
    \frac{\partial z^*}{\partial (\cdot)} = (I - \tilde{\phi}^\prime \tilde{D}_{K3})^{-1}\tilde{\phi}^\prime 
    \frac{\partial (\tilde{C}_{K2} \ \xi + \tilde{D}_{K3} \ z^* + \tilde{D}_{K4} \ y)}{\partial (\cdot)},
\end{equation*}
where \(\tilde{\phi}^\prime\) is the Jacobian of \(\tilde{\phi}\) evaluated at \(\tilde{C}_{K2} \xi + \tilde{D}_{K3} z^* + \tilde{D}_{K4} y\), and \(z^* = \tilde{\phi}(\tilde{C}_{K2} \xi + \tilde{D}_{K3} z^* + \tilde{D}_{K4} y)\).
A sufficient condition for the existence of this inverse is for \(\tilde{\phi}\) to be slope restricted in [-1, 1].
Then, \(\|\tilde{\phi}^\prime(v)\|_\Lambda \leq 1\) and \(\|\tilde{D}_{K3}\|_\Lambda < 1\), so \(I - \tilde{\phi}^\prime(v) \tilde{D}_{K3}\) is invertible for any \(v\).

See Alg. \ref{alg:train_lin_controller} for the procedure to train the controller for a linear plant.
During training, after each gradient step, the new parameters \(\hat{\theta}\) are projected to the stabilizing set \(\Theta\) and parameters \(\tilde{\theta}\) are recovered as in \eqref{eq:recovery}.
This guarantees that the controller used at all points during training exponentially stabilizes the origin of the closed loop system. 

\begin{algorithm} 
\caption{Stabilizing RL algorithm for linear plants}
    \begin{algorithmic} 
    \State \(\hat{\theta} \gets \text{random in }\ \Theta\)
    \While{not converged}
        \State \(\hat{\theta}^\prime \gets \text{gradient step from}\ \hat{\theta}\)
        \State \(\hat{\theta} \gets \arg\min_{\hat{\theta}} \|\hat{\theta}-\hat{\theta}^\prime\|_F \ \ \text{s.t.} \ \ \mathrm{LMI}(\hat{\theta})\)
    \EndWhile
    \State \(\tilde{\theta} \gets f(\hat{\theta})\) \Comment{Recover \(\tilde{\theta}\)}
    \end{algorithmic} \label{alg:train_lin_controller}
\end{algorithm}

\section{Partially Observed Linear Systems with Sector-Bounded Nonlinearity} \label{sec:non-lin-LTI}

\subsection{Problem Formulation}

In this section, we consider a nonlinear plant \(F_u(G, \Delta)\) which can be described as an interconnection of an LTI system \(G\) and a static sector bounded nonlinearity \(\Delta : \ell_{2e}^{n_\Delta} \rightarrow \ell_{2e}^{n_\Delta}\).
Sector bounded nonlinearities can be used to represent various types of nonlinear functions, e.g., activation functions as described in the previous sections, saturation functions~\cite{hindi1998analysis}, and model uncertainties \cite{buch2021robust}, which in turn enables us to consider plants modeled as neural networks.

The plant $F_u(G,\Delta)$ is given by the following equations 
\begin{align} 
&G \left\{\begin{array}{ll}
x(k+1) &= A_G \ x(k) \ + B_{G1} \ q^\prime(k) + B_{G2} \ u(k) \\
y(k) &= C_{G1} \ x(k) \\
p(k) &= C_{G2} \ x(k) + D_{G3} \ q^\prime(k)
\end{array}\right. \notag \\
&\hspace{0.87cm} q^\prime(k) \hspace{0.85cm} = \Delta(p(k)) \label{eq:plant-nonlin}
\end{align}
with \(p(k) \in \R^{n_\Delta}\) being the input to the nonlinearity and \(q^\prime(k) \in \R^{n_\Delta}\) being the output.

We assume \(\Delta\) operates element-wise with the \(i\)-th component being sector bounded by \([\alpha_{\Delta,i}, \beta_{\Delta,i}]\). 
Assume \(D_{G3}\) is such that the plant is well-posed.

The procedure to arrive at a stability condition now parallels that for the LTI plant. We first loop transform the plant so that the transformed nonlinearity \(\tilde{\Delta}\) is sector bounded by \([-1_{n_\Delta \times 1}, 1_{n_\Delta \times 1}]\).
The transformed parameters are denoted by a tilde. Note that some parameters remain the same under loop transformation.
The controller \(\pi_\theta\) has the same model as in \eqref{eq:PK_def}, and undergoes loop transformation as before.
Then, the feedback system of \(\pi_\theta\) and \(F_u(G, \Delta)\) is of the form:
\begin{subequations}\label{eq:nonlin-feedback_nominal}
\begin{align}
    \zeta(k+1) &= \mathcal{A} \ \zeta(k) + \mathcal{B} \ t(k)\\
s(k) &= \ \mathcal{C} \ \zeta(k) + \mathcal{D} \ t(k) \\
t(k) & = \tilde{\psi}(s(k)) \label{eq:nonlin-implicit}
\end{align}
\end{subequations}
where
 $\zeta = [x^\top, \ \xi^\top] ^\top$ is the stacked state, 
\(\tilde{\psi} = [\tilde{\Delta}^\top, \tilde{\phi}^\top]^\top\) is the stacked nonlinearity, \(s(k) = [p(k)^\top, \ v(k)^\top]^\top\) is the stacked  nonlinearity input, and \(t(k) = [q(k)^\top, \ z(k)^\top]^\top\) is the stacked nonlinearity output. 
Since both \(\tilde{\phi}\) and \(\tilde{\Delta}\) are sector bounded by \([-1, 1]\), so is \(\tilde{\psi}\).

\subsection{Stability Condition}

The gathered system \eqref{eq:nonlin-feedback_nominal} satisfies the same properties as the gathered system \eqref{eq:feedback_nominal}.
Therefore the condition from Lemma \ref{lemma:lin-plant-stability} applies to the feedback system of the sector bounded plant and REN controller as well. 
Then, decomposing \(P\) and defining \(\mathcal{Y}\) as before, \eqref{eq:big_lyap_cond_multiply2} represents an equivalent stability condition in terms of \(\mathcal{A}, \mathcal{B}, \mathcal{C}, \mathcal{D}\) matrices from \eqref{eq:nonlin-feedback_nominal}.


Decompose \(\Lambda\) as \(\mathrm{diag}(\Lambda_\Delta, \Lambda_\phi)\) where \(\Lambda_\Delta \in \mathbb{D}^{n_\Delta}_{++}, \Lambda_\phi \in \mathbb{D}^{n_\phi}_{++}\).
Then \(N, \hat{N}_{12}, \hat{N}_{21}, \hat{D}_{K4}\), and \(\hat{D}_{K3}\) can be defined paralleling the definitions for the LTI plant, using \(\Lambda_\phi\) instead of \(\Lambda\).
Expanding the terms shows that, for a fixed \(\Lambda_\Delta\), LMI \eqref{eq:big_lyap_cond_multiply2} is convex in \(\hat{\theta} = (X, Y, N, \Lambda_\phi, \tilde{D}_{K1}, \hat{N}_{12}, \hat{N}_{21}, \hat{D}_{K4}, \hat{D}_{K3})\).
The stabilizing set of parameters is 
\begin{equation}
    \Theta = \left\{\hat{\theta} : \exists \Lambda_\Delta \in \mathbb{D}_{++}^{n_\Delta}\ \ \text{LMI$(\hat{\theta}, \Lambda_\Delta)$ \ holds}\right\}.
\end{equation}
For any \(\hat{\theta} \in \Theta\), a corresponding \(\tilde{\theta}\) can be recovered as before.

\begin{theorem}
Consider the feedback system of plant \(F_u(G, \Delta)\) in \eqref{eq:plant-nonlin} and well-posed controller \(\pi_{\tilde{\theta}}\) in \eqref{eq:tilde-controller}.
Let a rate \(\rho \in [0, 1)\) be given.
If there exist \(\hat{\theta}\) and \(\Lambda_\Delta \in \mathbb{D}_{++}^{n_\Delta}\) such that \(\mathrm{LMI}(\hat{\theta}, \Lambda_\Delta)\) holds, then there exist parameters \(\tilde{\theta}\) for \(\pi_{\tilde{\theta}}\) such that the origin is exponentially stable with rate \(\rho\).
\end{theorem}

The proof is the same as for Theorem \ref{thm:stability-lin}. 

\subsection{Stabilizing Reinforcement Learning Problem}

The reinforcement learning procedure for the nonlinear plant, summarized in Alg.~\ref{alg:train-nonlin-controller}, follows closely as for the linear plant. 
However, since \(\mathrm{LMI}(\hat{\theta}, \Lambda_\Delta)\) is not convex in both \(\hat{\theta}\) and \(\Lambda_\Delta\), we split the projection step of \(\hat{\theta}^\prime\) onto the stabilizing set of parameters into two convex problems. 
The first problem fixes \(\Lambda_\Delta\) and projects \(\hat{\theta}^\prime\) onto \(\left\{\hat{\theta}\ :\ \mathrm{LMI}(\hat{\theta}, \Lambda_\Delta) \ \text{holds}\right\}\). 
Then, another convex problem is solved with \(\Lambda_\Delta\) as the decision variable and \(\hat{\theta}\) fixed to maximize the feasibility of \(\mathrm{LMI}(\hat{\theta}, \Lambda_\Delta)\) \cite[Remark 4]{yin2021backreach}.
As in Alg.~\ref{alg:train_lin_controller}, this guarantees exponential stability of the origin at all times during the training process.

\begin{algorithm}
\caption{Stabilizing RL algorithm for nonlinear plants}
    \begin{algorithmic}
    \State \(\Lambda_\Delta \gets I\) 
    \State \(\hat{\theta} \gets \text{random in }\ \Theta\)
    \While{not converged}
        \State \(\hat{\theta}^\prime \gets \text{gradient step from}\ \hat{\theta}\)
        \State \(\hat{\theta} \gets \arg\min_{\hat{\theta}} \|\hat{\theta}-\hat{\theta}^\prime\|_F \ \ \text{s.t.} \ \ \mathrm{LMI}(\hat{\theta}, \Lambda_\Delta)\)
        \State \(\Lambda_\Delta \gets \arg\max_{\Lambda_\Delta^\prime} 0 \ \  \text{s.t.} \  \mathrm{LMI}(\hat{\theta}, \Lambda_\Delta^\prime)\) \Comment{Re-center \(\Lambda_\Delta\)}
    \EndWhile
    \State \(\tilde{\theta} \gets f(\hat{\theta})\) \Comment{Recover \(\tilde{\theta}\)}
    \end{algorithmic}
    \label{alg:train-nonlin-controller}
\end{algorithm}

\section{Examples}

When the controller nonlinearity is slope restricted and passes through the origin, \(\tilde{\phi}(\tilde{C}_{K2}\xi(k) + \tilde{D}_{K3}z + \tilde{D}_{K4}y(k))\) is a contraction map with respect to \(z\),
as shown in the proof of Lemma \ref{lemma:well-posedness}. 
Therefore, we find the fixed point \(z(k)\) of this map using iteration methods \cite{bai2019deep}.
All code used can be found at \url{https://github.com/neelayjunnarkar/stabilizing-ren}.

\subsection{Training a controller for an inverted pendulum}

In this example we consider the following nonlinear inverted pendulum model:
\begin{align}
    & \bmat{x_1(k+1) \\ x_2(k+1)}  = \bmat{1 & \delta \\ \frac{g\delta}{\ell} & 1 - \frac{\delta\mu}{m\ell^2}}
        \bmat{x_1(k) \\ x_2(k)} \nonumber \\ 
    &\hspace{0.8cm} + \bmat{0 \\ -\frac{g\delta}{\ell}} \left(x_1(k) - \sin\left(x_1(k)\right)\right)
        + \bmat{0 \\ \frac{\delta}{m\ell^2}} u(k), \nonumber \\
    &\hspace{1.2cm}y(k)  = x_1(k), \label{eq:inv_pend_model}
\end{align}
where \(x_1(k)\) is the angular deviation from the pendulum being inverted, \(x_2(k)\) is the angular velocity of the pendulum, \(\delta = 0.02 \) s is the discretization sampling time, \(u(k)\) is the control input torque, \(m = 0.15\) kg is the mass, \(\ell = 0.5\) m is the pendulum length, and \(\mu = 0.5\) Nms$/$rad is the coefficient of friction.
We let \(\Delta(v) = v - \sin(v)\) be the nonlinearity.

We compute the reward for a rollout of length \(T\) as \(\sum_{k=0}^T \left(4 - u(k)^2\right)\), and train the controller \(\pi_{\tilde{\theta}}\) following the procedure in
Alg.~\ref{alg:train-nonlin-controller} using proximal policy optimization in the RLLib framework with PyTorch. The bias term in the reward function is to ensure that reward is nonnegative, and thus trajectories are not rewarded for ending early.
We enforce an exponential stability rate of \(\rho = 0.999\), set \(\phi = \tanh\), and use \(n_\xi = 2\) and \(n_\phi = 4\).
Initial conditions  \((x_1(0), x_2(0))\) are sampled from \([-0.3 \ \pi, 0.3 \ \pi] \ \text{rad} \times[-0.8, 0.8] \ \text{rad$/$s}\).
The plant nonlinearity \(\Delta(v)\) is globally slope restricted in \([0, 2]\). 
The maximum rollout length is 200 steps.
We also train a RNN controller from \cite{gu2021recurrent} with larger hidden size \(n_\phi = 8\) and using the global sector boun \([0, 1.213]\) for \(\Delta(v)\), but otherwise the same hyperparameters as for the REN controller.

Fig.~\ref{fig:ren-vs-rnn-old} compares reward versus number of samples taken from the plant model.
Due to the expressiveness of the additional parameter and the improved convexification procedure, the REN method we present in this paper achieves slightly better performance than the RNN controller from \cite{gu2021recurrent} while having a smaller model.
\begin{figure}[b]
    \centering
    \includegraphics[width=\linewidth,trim={0 0 0 0.9cm},clip]{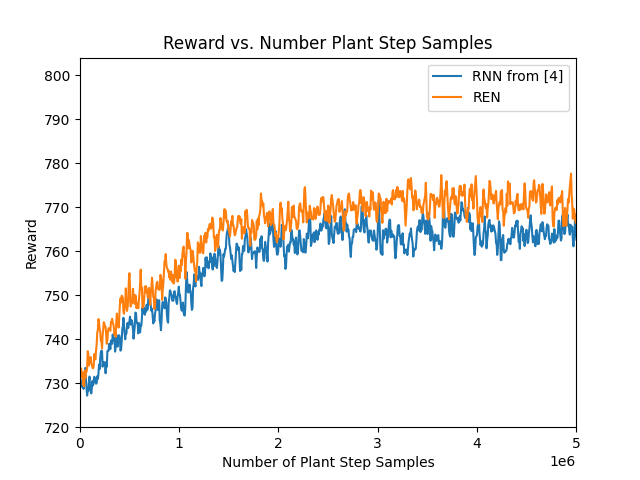}
    \caption{Reward during training of our REN-based controller with \(n_\xi = 2,\ n_\phi = 4\) and RNN controller from \cite{gu2021recurrent} with \(n_\xi = 2,\ n_\phi = 8\).
    The additional expressiveness of the REN-based controller and convexification procedure allows for slightly better performance with a smaller model.
    Note that stability is guaranteed at all points in training in both methods.
    }
    \label{fig:ren-vs-rnn-old}
\end{figure}

\subsection{Stabilizing a Neural Network Plant Model}

Neural networks have seen significant use in the identification of nonlinear systems.
The following demonstrates that the controller we present can stabilize a neural network model.
Consider a plant of the form
\begin{subequations} \label{eq:nn-plant-model}
\begin{align}
    x(k+1) & = F(x(k), u(k)) \\
    y(k) & = C_1 x(k)
\end{align}
\end{subequations}
where \(C_1\) is assumed to be known, and \(F\) is the output of the following implicit neural network:
\begin{subequations} \label{eq:nn-plant-nn}
\begin{align}
    F(x, u) & = Ax + B_1 q^\prime + B_2 u \\
    q^\prime & = \Delta(C_2 x + D_3 q^\prime). 
\end{align}
\end{subequations}
This implicit model encompasses a large variety of neural networks, including dense feedforward networks and convolutional neural networks \cite{elghaoui2020}. 
Expanding \(F(x(k), u(k))\) shows this plant model is of the desired plant form in \eqref{eq:plant-nonlin}.

We train a model of the form in \eqref{eq:nn-plant-nn} on the inverted pendulum with \(\Delta = \tanh\), \(q^\prime\) having size 2, the Adam optimizer, mean-square-error loss, and a learning rate of \(10^{-4}\). 
The training and test datasets are created by randomly sampling initial states in \([-\pi, \pi]\ \text{rad} \times [-8, 8]\ \text{rad$/$s}\), control inputs in \([-2, 2]\) Nm, and computing the true next state by simulating the true inverted pendulum model with the control.
After each gradient step, if \(\sigma = \|D_3\|_\Lambda \geq 1\), we set \(D_3 \gets \frac{\alpha}{\sigma}D_3\) with \(\alpha = 0.99\) to ensure that the implicit model remains well-posed.
The trained model achieves a mean-squared-error loss on the order of \(10^{-7}\). 

Next, we train a controller \(\pi_{\tilde{\theta}}\) similarly to the previous section, but assume that \(\Delta\) is incrementally sector bounded in \([0, 1]\) (a global incremental sector bound for \(\tanh\)), and set hyperparameters \(\rho = 0.9, n_\xi = 2\), and \(n_\phi = 8\).
We use the learned plant parameters in the projection step and the real plant model in computing trajectories. 
We also use a reward function that incentivizes stability since \(\pi_{\tilde{\theta}}\) is guaranteed to stabilize the learned plant model but not necessarily the true plant model.
The reward for a trajectory of length \(T\) is computed as \(\sum_{k=0}^T \left( 5 - x(k)^\top Q x(k) - u(k)^\top R u(k) \right)\) with \(Q = \mathrm{diag}(1, 0.1)\) and \(R = 0.01\).
Evaluation initial conditions are sampled in \([-0.6 \ \pi, 0.6 \ \pi] \ \text{rad} \times[-2, 2] \ \text{rad/s}\).
As before, the bias in the reward is experimentally determined to ensure reward is nonnegative, and thus trajectories are not rewarded for halting early. 

For comparison, we train another controller on the stability-incentivizing reward function, but with access to the true model in \eqref{eq:inv_pend_model} for the projection step. 
We assume the plant nonlinearity, \(\Delta(v) = v - \sin(v)\), is incrementally sector-bounded in \([0, 2]\), which is satisfied globally.
Fig.~\ref{fig:phase-portraits} plots phase portraits of both controllers, grid-sampling forty-nine initial conditions over \([-\pi, \pi] \times [-8, 8]\).
Both controllers achieve convergence to the origin from all tested points. 
Fig.~\ref{fig:learned-plant-vs-true-plant} compares the reward achieved by the two controllers.
The controller which guarantees stability of the learned, inexact model experimentally shows stability of the true model and achieves comparable performance to the controller which guarantees stability of the true model.

\begin{figure}[h]
    \centering
    \includegraphics[width=\linewidth]{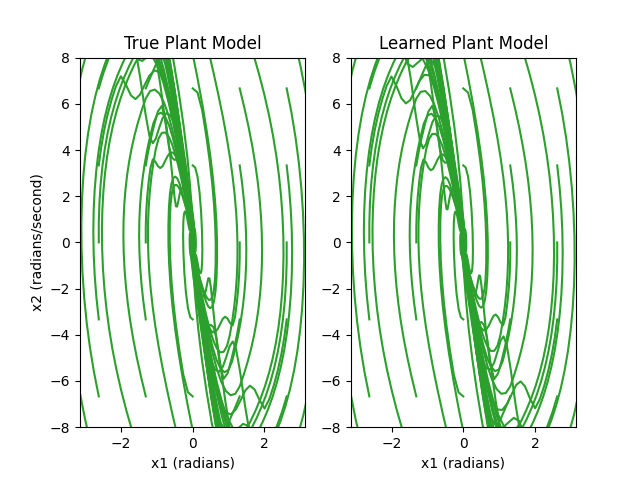}
    \caption{Phase portraits of the controller with stability guarantee of the true plant model vs the controller with stability guarantee of the learned plant model. Both controllers achieve convergence to the origin with all trajectories.
    }
    \label{fig:phase-portraits}
\end{figure}

\begin{figure}[h]
    \centering
    \includegraphics[width=\linewidth]{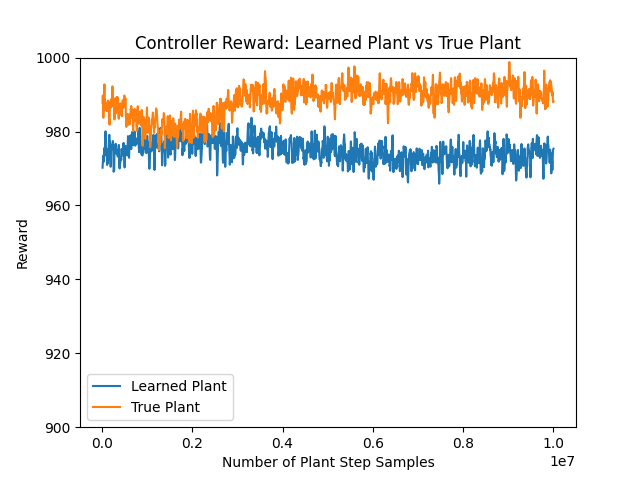}
    \caption{Reward of controller guaranteeing stability of learned plant model vs true plant model: The controller that only guarantees stability of the learned plant model achieves reward near that of the controller which has access to the true model parameters.
    }
    \label{fig:learned-plant-vs-true-plant}
\end{figure}

\section{Conclusion}

In this paper we presented a parameterization of dynamic controllers with sector bounded nonlinearities to guarantee exponential stability of the origin for partially-observed plants with sector bounded nonlinearities. 
This set of plants includes plants modeled by neural networks.
We demonstrated a controller synthesis method in a reinforcement learning context where an arbitrary gradient-based training algorithm can be used with the addition of a projection step that guarantees stability, both during training and of the trained controller.
The benefits of the expressiveness gained by this model and convexification method are demonstrated on a nonlinear inverted pendulum model.

\bibliographystyle{IEEEtran}
\bibliography{IEEEabrv,ref}

\end{document}